\documentclass[letterpaper, 10 pt, conference]{ieeeconf}  

\IEEEoverridecommandlockouts                              
\makeatletter
\let\NAT@parse\undefined
\makeatother
  
\usepackage{amsmath}   
\usepackage{amssymb}   
 \usepackage{amsthm}    
\usepackage{latexsym}  

\usepackage[authoryear,round,longnamesfirst]{natbib}

\usepackage[capitalize]{cleveref}
\usepackage{graphicx}
\usepackage{xurl}
\usepackage{mpcsymbols}
\usepackage{catchfilebetweentags}
\usepackage{flushend}

\AddToHook{env/proposition/begin}{\crefalias{theorem}{proposition}}
\AddToHook{env/lemma/begin}{\crefalias{theorem}{lemma}}
\AddToHook{env/corollary/begin}{\crefalias{theorem}{corollary}}
\AddToHook{env/definition/begin}{\crefalias{theorem}{definition}}
\AddToHook{env/remark/begin}{\crefalias{theorem}{remark}}

\newtheorem{theorem}{Theorem}
\newtheorem{lemma}[theorem]{Lemma}
\newtheorem{proposition}[theorem]{Proposition}

\theoremstyle{definition}
\newtheorem{definition}[theorem]{Definition}
\newtheorem{remark}[theorem]{Remark}
\newcommand{\oM}{\overline{M}}

\renewcommand{\paragraph}[1]{\textbf{#1}}

\title{On asymptotic stability of the time-varying Kalman filter for unstabilizable linear systems: an optimization perspective}

\author{James B. Rawlings, Titus Quah, and Matthias A. M{\"u}ller%
\thanks{J.B. Rawlings and T. Quah are with Department of Chemical Engineering,
University of California, Santa Barbara, CA 93106}%
\thanks{M.A. M{\"u}ller is with Leibniz University Hannover, Institute of Automatic Control, 30167 Hannover, Germany}%
}

\begin{document}

\maketitle
\thispagestyle{empty}
\pagestyle{empty}

\begin{abstract}
This paper establishes the necessary and sufficient conditions for asymptotic stability of the time-varying Kalman filter applied to a linear time invariant system with semidefinite initial state covariance and positive definite process and measurement noise.  Rather than analyze the discrete Riccati equation as in the classic literature, the equivalent state smoothing optimization problem is stated and all results are established using properties of this optimization problem. A Lyapunov-like function, termed a modified $Q$-function is derived and used for this analysis.
This optimization approach removes the need for the classic but cumbersome Riccati iteration algebra and provides better generalization and application for nonlinear systems.
\end{abstract}

\section{Introduction, Motivation, and Literature Review}
State estimation of a linear system by the Kalman filter lies at the heart of control and signal processing, with two closely related questions running through the literature: whether the estimator error dynamics are stable and whether the corresponding Riccati recursion governing the error covariance converges. The foundational work of \citet{kalman:bucy:1961} in continuous time and \citet{deyst:price:1968} in discrete time addressed both questions under strong observability and controllability assumptions. Subsequent work progressively weakened these assumptions. For Riccati convergence, observability was relaxed to detectability by \citet{wonham:1968a} and \citet{caines:mayne:1970}, and controllability to stabilizability by \citet{kucera:1973} and \citet{hager:horowitz:1976}; the latter also established stability results for the associated time-varying filter. \citet{anderson:moore:1981} subsequently established exponential stability of the time-varying discrete-time Kalman filter under uniform detectability and stabilizability. For a time, detectability and stabilizability appeared to mark the boundary of both theories. However, the Riccati-convergence theory eventually weakened stabilizability. \citet{desouza:gevers:goodwin:1986} showed that stabilizability is not necessary, and \citet{callier:winkin:1995} ultimately gave necessary and sufficient conditions for convergence to the strong solution of the Riccati equation—the limit whose associated steady-state filter may have modes on the unit circle. This is where the two questions part ways: covariance convergence and stable estimator error dynamics are closely related, but they are not the same question.

This distinction matters because convergence of the covariance is not convergence of the estimator. When the Riccati recursion converges to a stabilizing solution, the distinction is hidden: the limiting steady-state Kalman filter is exponentially stable, and the time-varying estimation error decays exponentially as well. However, when the covariance converges to a strong solution, the two properties need no longer coincide. The covariance may converge to a strong solution while the time-varying estimation error does not decay. More surprisingly, the time-varying Kalman filter may be asymptotically stable even though the steady-state filter to which it converges is not (see Exercise 4.18 in \citet{rawlings:mayne:diehl:2020}). Thus a Riccati convergence theorem, by itself, does not answer the question we address here: under precisely what conditions is the time-varying Kalman filter asymptotically stable?

We consider the linear, discrete, time invariant system
\begin{equation}
x^+  = Ax + Gw \qquad y = Cx + v
\label{eq:linsys}
\end{equation}
with random initial condition $x(0) \sim N(\ox_0, \Sigma_0)$ and noise
sequences $w(k) \sim N(0, Q)$ and $v(k) \sim N(0,R)$, each independent and
identically distributed in time, with $x(0)$, $w(\cdot)$, and $v(\cdot)$
mutually independent. We
assume throughout that $Q > 0, R > 0, \Sigma_0 \geq 0$. We shall consider more
restrictive assumptions on $\Sigma_0$ as we progress.
The time-varying Kalman filter estimate $\xhat(k)$, the one-step-ahead
prediction of $x(k)$ conditioned on the measurements $y(0),\ldots,y(k-1)$ (so
that $\xhat(0)=\ox_0$), has estimation error $\ehat(k) \eqbyd x(k)-\xhat(k)$
with covariance $\Sigma(k)$ given by the standard discrete Riccati iteration
(DRE)
\begin{align}
\Sigma(k+1) &= GQG' + A \Sigma(k) A' - \notag \\
&\phantom{=}  A \Sigma(k)C'(C\Sigma(k)C' + R)^{-1}C\Sigma(k) A'
\label{eq:dre}
\end{align}
with $\Sigma(0) = \Sigma_0, \, k\geq 0$.
We also consider system \cref{eq:linsys} in an orthogonal stabilizability canonical form
\begin{align}
\begin{bmatrix}x_1 \\ x_2 \end{bmatrix}^+ &=
\begin{bmatrix} A_1 & A_{12} \\ 0 & A_2 \end{bmatrix}
\begin{bmatrix}x_1 \\ x_2 \end{bmatrix} +
\begin{bmatrix} G_1 \\ 0 \end{bmatrix} w
\label{eq:contrb} \\
y &= \begin{bmatrix} C_1 & C_2 \end{bmatrix}
\begin{bmatrix}x_1 \\ x_2 \end{bmatrix} +  v \notag
\end{align}
in which $(A_1,G_1)$ is stabilizable, and $A_2$ is
completely unstable; there are $n_2$ unstabilizable modes of the system.

Two conditions on the system and prior govern the result.
\begin{enumerate}
\item[C1:] $(A,C)$ is detectable.
\item[C2:] The intersection of the null space of $\Sigma_0$  and the 
 unstable, uncontrollable subspace  of $(A,G)$ is zero.
\end{enumerate}
Our main result is the following.
\begin{proposition}[Time-varying estimator stability]
\label{prop:tvkf}
The time-varying linear optimal estimator is asymptotically stable if and only if
C1 and C2 hold.
\end{proposition}
We establish this equivalence from the optimization problem the time-varying Kalman filter solves, not from the Riccati recursion. 
The tool is a Lyapunov-like $Q$-function, adapted from \citet{allan:rawlings:2019b}. Their construction assumes the system is stabilizable and the prior positive definite, but the systems we consider are unstabilizable, with only semidefinite priors. Without stabilizability, the original $Q$-function loses the upper bound, so we introduce a modified $Q$-function which is not upper bounded, but instead relies on the filter's linearity to recover global asymptotic stability. Sufficiency of C1 and C2 follows this way; necessity of detectability is known \citep{allan:rawlings:teel:2021}, and necessity of C2 is, so far as we know, new.

The condition C2 is also what separates this result from the covariance convergence theory. \citet{callier:winkin:1995} obtain covariance convergence under detectability and a strictly weaker prior condition---one that constrains only the antistable uncontrollable modes, those strictly outside the unit circle. Estimate error stability requires C2, which constrains the unit-circle modes as well. On a system whose prior misses a unit-circle mode, the covariance still converges to its strong-solution limit while the estimate error does not, so a discrete-time covariance theorem does not provide strong enough conditions for estimator stability. We have not found this characterization in the state estimation literature, in either continuous or discrete time. The optimization-based stability arguments of \citet{hager:horowitz:1976} provide useful precedent for the approach taken here and, together with the present result, suggest that the optimization problem itself is a natural framework for studying estimator stability. This viewpoint extends naturally to nonlinear optimization-based estimators, whose analysis provided the original motivation for this study.

We prove asymptotic stability here; the exponential case is deferred to a later paper. The reader should be aware that this asymptotic stability is not robust. For an unstabilizable system with an uncontrollable mode on the unit circle, a bounded disturbance that decays to zero can still cause the estimate error to diverge, even under perfect measurement as shown in Example 2 of \cref{sec:examples}. Robustness requires exponential stability, which holds once the uncontrollable unit-circle modes are excluded---the subject of a subsequent paper.

\section{Discussion of analysis tools}

\begin{remark}[Lyapunov functions]
Note that the popular time-varying Lyapunov function candidate for estimate
error $V_k(\ehat) \eqbyd \ehat'\Sigma(k)^{-1} \ehat$ introduced by \cite{deyst:price:1968} for observable and controllable
systems is not a Lyapunov function for the case
considered in \cref{prop:tvkf} as we demonstrate in the Example 1 of \cref{sec:examples}.

When generalizing to detectable and stabilizable systems,
\cite{anderson:moore:1981} briefly consider the candidate\footnote{Note:
  $\Sigma(k)$ and not $\Sigma(k)^{-1}$ as in   \citep{deyst:price:1968}.}
$V_k(\ehat) \eqbyd  \ehat'\Sigma(k) \ehat$ and 
comment that this candidate ``is not necessarily 
positive definite, and in fact it is easy to construct examples where it fails
to be positive definite; another difficulty is that the monotone decreasing
property of $V$ along trajectories is not strict.''
\end{remark}

Deprived of Lyapunov functions for general estimation problems, previous
stability analysis has generally focused instead on properties of the DRE
itself, such as conditions that ensure monotonicity in the Riccati iteration.
But these approaches lead to complex algebraic manipulations and conservative
restrictions on the initial variance $\Sigma_0$ that we wish to remove.  In the
following, we avoid arguments based on algebraic DRE manipulations wherever
possible because these have no obvious extension to nonlinear systems. In
contrast, analysis based instead on the estimator's optimization problem often
is extendable to nonlinear systems.

\begin{remark}[Lyapunov-like function]
Without a Lyapunov function, we might try instead for a $Q$-function
as in \cite{allan:rawlings:2019b}. For stabilizable systems, that approach
readily establishes exponential stability of the estimator in exactly the same
fashion as a Lyapunov function.  But for unstabilizable systems, we do not have
the required upper bound for the $Q$-function.  We address that issue in
\cref{sec:Quns}. 
\end{remark}

 \subsection{Optimization problems.}
As a general analysis approach we find it convenient to study the optimization problems that
generate identical solutions to the Kalman filter, i.e., the famous (or infamous)
connection of optimal statistical estimation and recursive least squares.
For example, the time-varying Kalman filter with initial variance $\Sigma_0 > 0$ has a
one-to-one correspondence with the following optimization problem
\cite[pp.287-288]{rawlings:mayne:diehl:2020} cf. \cite[pp.138-139]{sorenson:1970b}, \cite[pp.205-207]{jazwinski:1970}, \cite[p.227]{anderson:1973}.
\begin{gather*}
\min_{\chi(0), \omegaseq_T} V_T \eqbyd \ell_x(\chi(0)-\ox_0)
 + \sum_{k=0}^{T-1} \ell(\omega(k), \nu(k))\\
\text{subject to~} \chi^+ = A \chi + G \omega, \quad y = C\chi + \nu\\
\ell_x(e) \eqbyd \tfrac12 \norm{e}^2_{\Sigma_0^{-1}},\,\quad\ell(w,v) \eqbyd \tfrac12 (\norm{w}^2_{Q^{-1}} + \norm{v}^2_{R^{-1}})
\end{gather*}
$y(k),\,k = 0, 1, \ldots, T-1$ are the measurements, $\ox_0$ is the prior
estimate of the initial state, and
$\omegaseq_T \eqbyd (\omega(0), \ldots, \omega(T-1))$.
We denote the optimal state estimate trajectory as $\xhat(j\mid T)$ with $0 \leq
j \leq T$, $T \geq 0$, and the final element of the optimal trajectory  as
$\xhat(T) \eqbyd \xhat(T\mid T)$. Similarly, $\what(j\mid T)$ and $\vhat(j\mid T)$ denote the optimal process- and measurement-noise trajectories, collected as
$\whatseq_T \eqbyd (\what(j \mid T))_{j\in[0:T-1]}$.

In the case of interest here, the initial covariance is only semidefinite,
the initial condition is distributed as a \textit{singular} normal, and the standard
objective function is not well defined since $\Sigma_0^{-1}$ does not exist. To
handle semidefinite initial covariance, first denote $\Sigma_0$'s SVD  and
pseudo-inverse $\Sigma_0^\dagger$ by 
\begin{align}
\Sigma_0 &= \begin{bmatrix} U_1 & U_2 \end{bmatrix} \begin{bmatrix} \tSigma_0 & \\
  & 0 \end{bmatrix} \begin{bmatrix} U_1' \\ U_2' \end{bmatrix} = U_1 \tSigma_0
  U_1' \label{eq:SigmaSVD}\\
\Sigma_0^\dagger &= U_1 \tSigma_0^{-1}U_1' \notag
\end{align}
where $U \eqbyd [\,U_1\ U_2\,] \in \bbR^{n\times n}$ is orthogonal,
$\tSigma_0 \in \bbR^{r \times r}$, and $r \geq 0$ is the rank of $\Sigma_0$.
Then consider $\Sigma_\rho \eqbyd \Sigma_0 + \rho I_n$ where $\Sigma_0 \geq 0$ and $\Sigma_\rho > 0$ for $\rho > 0$, and note that
\begin{equation*}
\lim_{\rho \searrow 0} \norm{x-\ox_0}^2_{\Sigma_\rho^{-1}} =
\begin{cases}
\norm{x-\ox_0}^2_{\Sigma_0^\dagger}, & \; U_2'(x - \ox_0) = 0 \\
\infty, & \; U_2'(x - \ox_0) \neq 0 
\end{cases} 
\end{equation*}   
Therefore, to obtain a well-defined, bounded objective function in the semidefinite limit, we redefine
\[
\ell_x(e) \eqbyd \tfrac12\norm{e}^2_{\Sigma_0^\dagger}
\]
and add the equality constraint.\footnote{Note that this required equality constraint was accidentally omitted in some of the early state estimation literature on the semidefinite case.}
\begin{lemma}[Semidefinite initial variance]
\label{lem:semiPT}
The time-varying Kalman filter with semidefinite initial variance $\Sigma_0 \geq
0$  has a one-to-one correspondence with the optimization problem $\bbP_T$ defined as
\begin{gather*}
\min_{\chi(0), \omegaseq_T} V_T \eqbyd 
\ell_x(\chi(0)-\ox_0)
 + \sum_{k=0}^{T-1} \ell(\omega(k),\nu(k))\\
\text{subject to~} \chi^+ = A \chi + G \omega, \quad y = C\chi + \nu,\\
U_2'(\chi(0)-\ox_0) = 0
\end{gather*}
\end{lemma}
At the edge cases, if $\Sigma_0>0$, the constraint with $U_2$ is deleted, and we
have the standard optimization problem. If $\Sigma_0=0$, the initial state term in
$V_T$ is identically zero, and the constraint becomes $\chi(0)=\ox_0$ since $U_2$ has
rank $n$.  Next we summarize some established properties of the solution
to $\bbP_T$ \cite[pp. 277--280]{rawlings:mayne:diehl:2020}. 
\begin{lemma}[Existence and Uniqueness of Solution to $\bbP_T$]
\label{lem:exist}
The optimal value $V_T^0$ of $\bbP_T$ exists and its minimizer is unique for all $A,G,C$ and $Q>0$, $R>0$, $\Sigma_0\geq0$.
\end{lemma}
\begin{proof}
Note that the cost function $V_T(\chi(0), \omegaseq_T)$ is a continuous function
of its arguments. Given $Q>0$, $V_T$ is radially unbounded in the $\omegaseq_T$
argument.  Since $\Sigma_0^\dagger$ is only semidefinite, consider the
variable transformation $\alpha = U' (\chi(0)-\ox_0)$, partitioned
conformably with $U = [\,U_1\ U_2\,]$ as $\alpha = (\alpha_1,\alpha_2)$, which gives 
\begin{equation*}
\norm{\chi(0)-\ox_0}^2_{\Sigma_0^\dagger} =
\norm{\alpha_1}^2_{\tSigma_0^{-1}}
\end{equation*}
and the constraint in $\bbP_T$ gives $\alpha_2 = 0$.  We then have that $V_T$ is
radially unbounded
in its free variable, $\alpha_1$, since $\tSigma_0 >0$, and $\alpha_2=0$. So
$V_T$ can be considered a continuous function of $(\alpha_1, \omegaseq_T)$, and
is radially unbounded in this argument, so the  solution to $\bbP_T$ exists for all
$T \geq 0$.
Moreover, in the variables $(\alpha_1,\omegaseq_T)$ the cost is a quadratic function whose second-degree terms include $\norm{\alpha_1}^2_{\tSigma_0^{-1}}$ and $\sum_j\norm{\omega(j)}^2_{Q^{-1}}$ with $\tSigma_0>0$ and $Q>0$, while the output terms contribute nonnegative second-degree terms since each residual is affine in $(\alpha_1,\omegaseq_T)$; hence $V_T$ is strictly convex and the minimizer is unique.
\end{proof}
Note that we have not made any assumptions about detectability or
stabilizability for existence of solution to problem $\bbP_T$.

\subsection{Linear systems}
As we specialize to estimation results that hold only for linear systems, we
finally resort to using properties of linear systems that do \textit{not} extend
in any obvious way to nonlinear systems. 
For the time-varying Kalman filter, we  have the following time-varying linear
system for the estimate error $\ehat$
\begin{equation}
\ehat(k+1) = (A-L(k)C)\ehat(k)
\label{eq:epssys}
\end{equation}
with initial condition $\ehat(0) = x(0)-\ox_0$ and the time-varying
estimator  gains 
\begin{equation}
L(k) \eqbyd A \Sigma(k) C'(C \Sigma(k) C' + R)^{-1}
\label{eq:Lk}
\end{equation}
The solution is then 
\begin{align}
\ehat(k) &= M(k) \ehat(0) \label{eq:epsk}\\
M(k) &= \big(A-L(k-1)C\big) \cdots \big(A-L(0)C\big),\, M(0) = I_n
\label{eq:Mk}
\end{align}

In the subsequent analysis of the estimator's \textit{nominal} stability properties, we assume that the measurements are coming from the nominal linear system without any disturbances, i.e.
\begin{equation}
y(k) = C A^k x(0) 
\label{eq:nommeas}
\end{equation}

Next we show that the cost function is uniformly bounded and converges.
\begin{lemma}[Uniformly bounded value function and convergence]
\label{lem:unibounded}
Given C2 and nominal measurements, the solution of $\bbP_T$ exists and its optimal
value is uniformly bounded above for all $T \geq 0$, i.e., 
there exists $c_v \geq 0$ such that the following holds
for all $x(0), \ox_0 \in \bbR^n$ and $T \geq 0$
\begin{equation*}
V_T^0 \leq c_v \norm{x(0) - \ox_0}^2
\end{equation*}
The sequence $\big(V_T^0\big)_{T \geq 0}$ converges, $V_\infty^0$ exists, and
\begin{equation}
\lim_{T \rightarrow \infty}  
\ell(\what(T \mid T+1), \vhat(T \mid T+1))
= 0
\label{eq:ellT}
\end{equation}
\end{lemma}
\begin{proof}
For $\Sigma_0 > 0$, a uniform upper bound for $V^0_T$ for all $T\geq 0$ is immediate using the feasible decision variables $\chi(0) = x(0), \omegaseq_T = 0$,  yielding $V_T^0 \leq \tfrac12\norm{x(0)-\ox_0}_{\Sigma_0^{-1}}^2$.
Using \cref{eq:quadbound} of \cref{sec:proofs} then gives $V_T^0 \leq c_v \norm{x(0)-\ox_0}^2$ with $c_v = \tfrac12\olambda(\Sigma_0^{-1})$.
So the first issue is to extend this easy upper bound to handle the semidefinite case $\Sigma_0 \geq 0$ under condition C2.  To do that we require extensive use of properties of linear systems, so we defer that part of the proof to \cref{sec:proofs}.

To establish convergence, note that the optimal solution of the estimation problem at $T+1$ gives feasible, but possibly suboptimal, decision variables at $T$, so we have the inequality
\begin{equation*}
  V_T^0 \leq V_{T+1}^0 -
\ell(\what(T \mid T+1), \vhat(T \mid T+1))
\end{equation*}
Therefore the sequence $\big(V_T^0\big)_{T \geq 0}$ is nondecreasing and bounded
above, and hence converges, so $V_\infty^0$ exists. Rearranging the inequality gives
\begin{equation*}
\ell(\what(T \mid T+1), \vhat(T \mid T+1))
\leq V_{T+1}^0 - V_T^0  
\end{equation*}
and the convergence of sequence $(V_T^0)$ implies that the right-hand side converges to zero. Since $\ell(\cdot) \geq 0$, so does the left-hand side, and the result is proven.
\end{proof}

\section{Unstabilizable systems and $Q$-functions}
\label{sec:Quns}

To establish stability of the optimal linear estimator for \textit{unstabilizable} systems, we require a modified $Q$-function adapted from \citet{allan:rawlings:2019b}. We will find the following form useful for this purpose.
\begin{definition}[Modified $Q$-function]
\label{def:modQ}
A function $Q(j~\mid~k)$ is a modified $Q$-function 
if there exist $\mc{K}_\infty$-functions $\mu_0,\mu_1, \mu_3$ such that  
\begin{align}
& \mathrel{\phantom{\leq}} Q(0 \mid k) \leq \mu_0(\norm{x(0)-\ox_0}) \label{eq:QunsInitUB} \\
\mu_1(\norm{x(j)-\xhat(j\mid k)}) &\leq Q(j \mid k) \label{eq:QunsLBUB} \\
Q(j+1\mid k) &\leq Q(j\mid k) - \mu_3(\norm{x(j)-\xhat(j\mid k)}) \label{eq:QunsDecrease}
\end{align}

for all $x(0)$, $\ox_0$, and $k\in\bbI_{\geq0}$, with $0\leq j\leq k$ in \cref{eq:QunsLBUB} and $0\leq j\leq k-1$ in \cref{eq:QunsDecrease}.
\end{definition}
It would seem at first glance that the loss of the upper bound in \cref{eq:QunsLBUB} would prevent concluding even asymptotic stability, but we will see that the \textit{linearity} of the time-varying Kalman filter saves the situation, and we can establish global asymptotic stability (GAS). 

\begin{proposition}[Modified $Q$-function and GAS]
\label{prop:modQgas}
Suppose C1 and C2 hold. If the optimal linear time-varying state estimator admits a modified
$Q$-function for which the horizon limits
$Q(j\mid\infty)\eqbyd \lim_{k\tends\infty}Q(j\mid k)$ exist for all $j\ge0$, then it is GAS.
\end{proposition}

\paragraph{Classical argument.}  To understand why this proposition might be valid,
consider first establishing the classical (weaker) two-part definition of GAS.
\begin{definition}[GAS estimation (classical)]
\label{def:classicalGAS}
An estimator is GAS (classical) if 
\begin{enumerate}
\item For every $\eps>0$, there exists $\delta>0$
such that for all $x(0)$ and $\ox_0$  satisfying $\norm{x(0) - \ox_0} \leq
\delta$ the following holds for all $k \geq 0$
\begin{equation*}
\norm{x(k)-\xhat(k)}  \leq \eps
\end{equation*}

\item 
For all $x(0)$ and $\ox_0$
\begin{equation*}
\lim_{k \tends \infty} \norm{x(k)-\xhat(k)}   = 0 
\end{equation*}
\end{enumerate}

\end{definition}

\begin{proof}[Proof (classical case)]
We can establish Uniform Global Stability, which is stronger than the local
stability (Lyapunov stability) required in classical GAS.  
From  repeated use of \cref{eq:QunsDecrease} we know that
$Q(j\mid k) \leq Q(0\mid k)$ for all $j\leq k, k \in \bbI_{\geq 0}$.  Therefore, applying $\mu_1^{-1}$ to
\cref{eq:QunsLBUB}
and substituting \cref{eq:QunsInitUB}, we have that
\begin{align*}
\norm{x(j)-\xhat(j\mid k)} &\leq \mu_1^{-1}(Q(j\mid k)) \leq \mu_1^{-1}(Q(0\mid k)) \\
 &\leq \mu_1^{-1}\circ \mu_0(\norm{x(0)-\ox_0})
\end{align*} 
for all $j \leq k$, $k \in \bbI_{\geq 0}$, $x(0), \ox_0 \in \bbR^n$, meeting
the definition of Uniform Global Stability, which implies local stability
\cite[Remark 7]{mcallister:rawlings:2020}. 

To address convergence, we require the limits of the state estimation solutions
as the horizon grows, established in \cref{prop:infhor} of \cref{sec:proofs}
under C1 and C2.
The limit $Q(j\mid\infty) \eqbyd \lim_{k\tends\infty}Q(j\mid k)$ exists for all
$j\geq0$ by hypothesis; by \cref{prop:infhor}.\ref*{it:zlim} and the
continuity of $\mu_1$, $\mu_3$, \cref{eq:QunsLBUB,eq:QunsDecrease} are
preserved in the limit.
\Cref{prop:infhor}.\ref*{it:xTT} further gives $\xhat(j\mid j) \tends \xhat(j\mid\infty)$ as $j\tends\infty$. With those facts in hand, we proceed as follows:
\cref{eq:QunsDecrease} implies that $Q(j\mid \infty)$ is a nonincreasing sequence that is
bounded below by zero, and therefore converges for all $x(0), \ox_0 \in \bbR^n$.
Rearranging \cref{eq:QunsDecrease} for $k=\infty$ gives
\begin{equation*}
\norm{x(j)-\xhat(j\mid \infty)} \leq \mu_3^{-1}(Q(j \mid \infty) - Q(j+1\mid \infty))
\end{equation*}
for all $j \geq 0$. 
Since $Q(j\mid \infty) - Q(j+1\mid \infty)$ converges to zero as $j \tends \infty$, we
have that $\xhat(j\mid \infty) \tends x(j)$ as $j \tends \infty$, and therefore
$\xhat(j \mid  j) \tends x(j)$ as $j \tends \infty$, and convergence has been established.
\end{proof}

To establish the stronger KL form of GAS, we require Uniform Global Attraction
in addition to the Uniform Global Stability established above
\cite[Proposition 19]{mcallister:rawlings:2020}. What we have 
shown so far is only Global Attraction.    But since the time-varying Kalman
filter is a \textit{linear} system, we can next readily establish that the global attraction
is also uniform and the time-varying Kalman filter is GAS.  

\begin{definition}[GAS estimation]
\label{def:GAS}
An estimator is GAS if  there exists a $\mc{KL}$-function $\beta(\cdot)$  such that 
\begin{equation*}
\norm{x(k)-\xhat(k)}  \leq \beta(\norm{x(0) - \ox_0}, k)
\end{equation*}
for all $x(0), \ox_0 \in \bbR^n$ and all $k \geq 0$.
\end{definition}

\begin{proof}[Proof (KL version)]
We first establish that $M(k)$ is uniformly bounded for all $k \geq 0$.  Consider first
the DRE \cref{eq:dre} with index $k = 0$. Since $\Sigma(0) \geq 0$ and $R>0$,
$(C\Sigma(0)C' +R)^{-1}$ is well defined and positive definite. We next note that
$\Sigma(0) - \Sigma(0)C'(C\Sigma(0)C' +R)^{-1}C\Sigma(0)$ remains positive
semidefinite because it is the variance of a state $x$ conditioned on a
measurement $Cx + v$ where $v$ has variance $R$. Multiplying from left and right
by $A$ and $A'$ and adding the semidefinite $GQG'$ shows that
$\Sigma(1)$ is also semidefinite.
Iterating this argument gives $\Sigma(k)$ is well defined and positive
semidefinite  for all $k \geq 0$.  Therefore $L(k)$ in \cref{eq:Lk} is well
defined, and therefore $M(k)$ in \cref{eq:Mk} is well defined for all $k \geq
0$.
Recall that $\ehat(k) = M(k) \ehat(0)$. 
Since we have shown that $\lim_{k\tends \infty} \ehat(k) = 0$
for all $\ehat(0)$, we can choose $\ehat(0)$ to be the unit vectors in
$\bbR^n$, $\ehat(0) = \ehat_i$, for $i =1, 2, \ldots, n$ to conclude that each column
of matrix $M(k)$ tends to zero and thus $\lim_{k\tends \infty} M(k) = 0$.
Therefore, choose an arbitrary $\eps>0$, and we have an integer $K(\eps)$ such that
$\norm{M(j)} \leq \eps $ for all $j \geq K(\eps)$.  To compute a uniform
upper bound for $M(k)$, note that
\begin{align*}
\sup_{j \geq 0} \norm{M(j)} &= \max (\max_{0
  \leq j \leq K-1} \norm{M(j)}, \sup_{j \geq K} \norm{M(j)})\\
 &\leq \max(\norm{M(j^*)},
\eps) \bydeq \oM
\end{align*} 
where $j^*$ is some integer $0 \leq j^* \leq K-1$. Since $M(j^*)$ and $\eps$
are some finite constants, $\oM$ is a finite uniform upper bound for
$\norm{M(k)}$ for all $k \geq 0$.

To establish GAS, define a function, $\alpha(\cdot)$
\begin{equation*}
\alpha(k) = \sup_{j \geq k}  \norm{M(j)}
\end{equation*}
which exists for all $k \geq 0$ because $\norm{M(j)}$ is bounded above for all $j \geq
0$, and converges to zero as $k \tends \infty$. 
We have that $\alpha(\cdot)$ is an $\mc{L}$-function because it is nonincreasing by construction and converges to
zero.  Then define the 
$\mc{KL}$-function $\beta(\cdot)$ as the product of  
$\mc{K}$- and $\mc{L}$-functions, $\beta(r,k) = r \alpha(k)$.
We then have that for all $\ehat(0) \in \bbR^n$ and $k \geq 0$,
$\norm{\ehat(k)} = \norm{M(k) \ehat(0)} \leq \norm{M(k)} \norm{\ehat(0)} \leq
\alpha(k) \norm{\ehat(0)} = \beta(\norm{\ehat(0)}, k)$
 and the time-varying Kalman filter is GAS. 
\end{proof}

Next we show that C1 and C2 are sufficient for existence of the modified $Q$-function and therefore GAS of the time-varying Kalman filter. 
\begin{proposition}
\label{prop:tvkfQuns}
The time-varying Kalman filter admits a modified $Q$-function if C1 and C2 hold. Moreover, the $\mc{K}_\infty$-functions $\mu_0(\cdot)$, $\mu_1(\cdot)$, $\mu_3(\cdot)$ may be taken as quadratic functions.
\end{proposition}

The construction of the $Q$-function for nonlinear systems is given in \cite{allan:rawlings:2019b}, but under the assumption of stabilizability and a positive definite initial penalty.  So here we can follow the same basic procedure while  restricting to linear systems but extending to unstabilizable systems with only positive semidefinite initial penalty and C2. The biggest change is the loss of the upper bounding function in \eqref{eq:QunsLBUB} that holds for stabilizable systems. The details are provided in \cref{sec:proofs}.  With these tools assembled we are ready to prove the main result.

\medskip
\noindent\paragraph{Proof of  \cref{prop:tvkf}. }
\begin{proof}
Taken together \cref{prop:modQgas,prop:tvkfQuns}
establish that the time-varying linear estimator is
asymptotically stable if C1 and C2 hold, establishing their sufficiency.
The hypotheses of \cref{prop:modQgas} hold for the $Q$-function constructed
in the proof of \cref{prop:tvkfQuns} in \cref{sec:proofs}: for each $j$,
$Q(j\mid k)$ is a continuous function of $\xhat(0\mid k), \what(0\mid k),
\ldots, \what(j-1\mid k)$, which converge as $k\tends\infty$ by
\cref{prop:infhor}.\ref*{it:zlim}, so the horizon limits $Q(j\mid\infty)$
exist.
To establish necessity of C1, note that \cite[Proposition 2.6]{allan:rawlings:teel:2021} establishes
that detectability (C1) is necessary for asymptotic stability, even for general
nonlinear systems, subsuming the linear case considered here.
Next we show that not C2 also implies that the time-varying Kalman filter is
not asymptotically stable. Consider an unstabilizable system partitioned into
orthogonal stabilizability canonical form as in \eqref{eq:contrb}. Not C2
implies there exists a witness $\xi = (0,\xi_2)$, $\xi_2 \neq 0$ lying in the intersection between the
null space of $\Sigma_0$ and the unstable, uncontrollable subspace of $(A,G)$. Choose $x(0)=0$ and
$\ox_0=-\xi$, so the prior error $x(0)-\ox_0$ is the witness itself. Note also that C1 and not C2 force $C\neq0$. Failure of C2 implies
$n_2\ge1$, while $C=0$ would make every mode of the nonempty block $A_2$
unobservable. Since the eigenvalues of $A_2$ lie on or outside the unit circle,
this contradicts C1. For the considered nominal case, the choice $x(0)=0$ gives $x(k)=0$ and $y(k)=0$ for all
$k$, so the estimate error is $x(k)-\xhat(k)=-\xhat(k)$ and the residuals of
$\mathbb{P}_T$ are
$\vhat(k\mid T)=-C\xhat(k\mid T)$. The estimator is asymptotically stable
only if $\xhat(k)\rightarrow0$. We show that the pinned component forces
$V_T^0$ to grow at least linearly, while $V_T^0$ is bounded above by the
cumulative squared estimate. It follows that $\xhat(k)$ has nonvanishing
average energy.

First, we show that at every horizon $T$, the unstabilizable component of the
optimizer's initial estimate $\xhat_2(0\mid T)$ remains bounded below in norm
by the witness $\xi$. The constraint of $\mathbb{P}_T$ gives
$U_2'(\xhat(0\mid T)-\ox_0)=0$ for all $T\geq0$, so the prior deviation
$\xhat(0\mid T)-\ox_0$ lies in the range of $\Sigma_0$. Since $\xi$ lies in
the null space of $\Sigma_0$, i.e. in the range of $U_2$, it is orthogonal to
that deviation, and
$\langle \xi,\xhat(0\mid T)\rangle=-\|\xi\|^2$. Only the second block
participates because $\xi$ has no stabilizable component, and Cauchy--Schwarz
gives

\begin{equation}
\|\xhat_2(0\mid T)\| \geq \|\xi\|
\label{eq:pin}
\end{equation}

for all $T\geq0$. Thus the constraint does not fix $\xhat_2(0\mid T)$, but it
prevents the optimizer from removing its component along the witness.

Next, we show that the pinned component forces $V_T^0$ to grow linearly. Since
C1 holds, let $V_\io(\tx)=\tfrac12\tx'P\tx$, $P\succ0$, and
$a_1,a_2,a_3,c_1,c_2>0$ be as in
\eqref{eq:iioss-bounds}--\eqref{eq:iioss-dec}. Applying
\eqref{eq:iioss-dec} along the optimal trajectory of $\mathbb{P}_T$, with
$\tx=\xhat(k\mid T)$, $\tw=\what(k\mid T)$, and $\ty=C\xhat(k\mid T)= -\vhat(k\mid T)$, summing over
$k=0,\ldots,T-1$, and discarding $V_\io(\xhat(T\mid T))\ge0$ gives

\begin{equation}
\begin{aligned}
a_3\sum_{k=0}^{T-1}\|\xhat(k\mid T)\|^2
&\leq V_\io(\xhat(0\mid T))\\
&\phantom{\leq} +\sum_{k=0}^{T-1}c_1\|\what(k\mid T)\|^2 +c_2\|\vhat(k\mid T)\|^2
\end{aligned}
\label{eq:iosssum}
\end{equation}

The right-hand side of \eqref{eq:iosssum} contains exactly the initial-error,
process-residual, and output-residual terms penalized in $\mathbb{P}_T$, but in
Euclidean rather than weighted norms. Since $\ox_0=-\xi$,
\eqref{eq:quadbound} with
$\|x+y\|^2\leq2(\|x\|^2+\|y\|^2)$ gives
$V_\io(\xhat(0\mid T))
\leq
\overline{\lambda}(P)
\bigl(\|\xi\|^2+\|\xhat(0\mid T)-\ox_0\|^2\bigr)$,
\eqref{eq:rangebound} converts
$\|\xhat(0\mid T)-\ox_0\|^2
\leq
\overline{\lambda}(\Sigma_0)
\|\xhat(0\mid T)-\ox_0\|_{\Sigma_0^\dagger}^2$,
and \eqref{eq:quadbound} converts
$\|\what\|^2
\leq
\|\what\|_{Q^{-1}}^2/\ulambda(Q^{-1})$
and likewise for $\vhat$ with $R^{-1}$. Hence the right-hand side of
\eqref{eq:iosssum} is bounded by the optimal cost:

\begin{equation}
\begin{aligned}
V_\io(\xhat(0\mid T))
&+
\sum_{k=0}^{T-1}c_1\|\what(k\mid T)\|^2 +c_2\|\vhat(k\mid T)\|^2 \\
&\leq
\overline{\lambda}(P)\|\xi\|^2
+
2c_\ell V_T^0
\end{aligned}
\label{eq:rhs-bound}
\end{equation}

where
\[
c_\ell
\eqbyd
\max\left\{
\overline{\lambda}(P)\overline{\lambda}(\Sigma_0),\;
\frac{c_1}{\ulambda(Q^{-1})},\;
\frac{c_2}{\ulambda(R^{-1})}
\right\}
>0
\]

On the left of \eqref{eq:iosssum}, each term is at least as large as the
contribution of the unstabilizable block, which evolves autonomously as
$\xhat_2(k+1\mid T)=A_2\xhat_2(k\mid T)$ so
$\xhat_2(k\mid T)=A_2^k\xhat_2(0\mid T)$. Every eigenvalue of $A_2$ has
modulus at least one, so applying \eqref{eq:gramian} with $B=A_2$ and using
\eqref{eq:pin} gives

\begin{equation}
\begin{aligned}
\sum_{k=0}^{T-1}\|\xhat(k\mid T)\|^2
&\geq
\sum_{k=0}^{T-1}
\|A_2^{\,k}\xhat_2(0\mid T)\|^2 \\
&\geq
c_B T\|\xhat_2(0\mid T)\|^2
\geq
c_B T\|\xi\|^2
\end{aligned}
\label{eq:lhs-bound}
\end{equation}

for $T\geq1$. Chaining \eqref{eq:lhs-bound} and \eqref{eq:rhs-bound} through
\eqref{eq:iosssum} gives
$a_3c_BT\|\xi\|^2
\leq
\overline{\lambda}(P)\|\xi\|^2+2c_\ell V_T^0$, and dividing by $T$ and the
positive denominator,

\begin{equation}
\liminf_{T\to\infty}\frac{V_T^0}{T}
\geq
c_g,
\qquad
c_g
\eqbyd
\frac{a_3c_B\|\xi\|^2}{2c_\ell}
>0
\label{eq:Vlinear}
\end{equation}

Now we upper bound $V_T^0$ with the cumulative squared estimate. Note that at
$T=0$ the estimator has no stages and $\xhat(0\mid0)=\ox_0$ is feasible at
zero cost, so $V_0^0=0$. The optimal solution at horizon $T$, extended by the
single stage $\omega(T)=0$, is feasible at horizon $T+1$: the constraint and
the first $T$ stages are unchanged and the cost increases by the one new stage
$\ell(0,\vhat(T\mid T))$. Since
$\xhat(T)=\xhat(T\mid T)$, that residual is $-C\xhat(T)$, so by optimality
of $V_{T+1}^0$ and \eqref{eq:quadbound},
$V_{T+1}^0
\leq
V_T^0+
(1/2)\|C\xhat(T)\|_{R^{-1}}^2
\leq
V_T^0+c_y\|\xhat(T)\|^2$
with
$c_y\eqbyd\|C\|^2/(2\ulambda(R))>0$, and summing from
$V_0^0=0$ gives

\begin{equation}
V_T^0
\leq
c_y\sum_{k=0}^{T-1}\|\xhat(k)\|^2
\label{eq:Venergy}
\end{equation}

Combining \eqref{eq:Vlinear} and \eqref{eq:Venergy} gives

\begin{equation}
\liminf_{T\to\infty}
\frac{1}{T}
\sum_{k=0}^{T-1}\|\xhat(k)\|^2
\geq
\frac{c_g}{c_y}
\label{eq:avgerr}
\end{equation}

i.e., the running average of $\|\xhat(k)\|^2$ is bounded away from zero. By
Cesàro convergence,
$\|\xhat(k)\|^2\rightarrow0
\Rightarrow
(1/T)\sum_{k=0}^{T-1}\|\xhat(k)\|^2\rightarrow0$.
However, \eqref{eq:avgerr} gives
$(1/T)\sum_{k=0}^{T-1}\|\xhat(k)\|^2\not\rightarrow0$, hence
$\xhat(k)\not\rightarrow0$. Since $x(k)=0$, the estimate error
$x(k)-\xhat(k)$ does not converge to zero, and C2 is therefore also
necessary for asymptotic stability, and we have established
\cref{prop:tvkf}.
\end{proof}

\section{Examples}
\label{sec:examples}
\subsection*{Example 1: $V_k(\cdot) = \norm{\cdot}_{\Sigma(k)^{-1}}^2$ is not a
  Lyapunov  function} 
The candidate function is not a Lyapunov function
for the systems of interest.  Let $A = \alpha I_n$, $C = O_n$, $G = O_n$,
$Q=I_n$, $R=I_n$, $\Sigma(0) = I_n$ with scalar $\alpha \in (0,1)$.   The system is
detectable and stabilizable 
with a positive definite initial covariance.   The solution of the DRE
gives $\Sigma(k) = \alpha^{2k}I_n$, and $L(k)=0$.   The estimate
error is $\ehat(k) = A^k \ehat(0) = \alpha^k \ehat(0)$. 
The candidate Lyapunov function is then 
\begin{equation*}
V_k(\ehat(k)) = \alpha^{2k}\ehat(0)'\alpha^{-2k}I_n \ehat(0) = \norm{\ehat(0)}^2
\end{equation*}
which is a constant function of time and does not satisfy any cost decrease while the estimate error
converges to zero exponentially fast.

\subsection*{Example 2:  An asymptotically stable Kalman filter is not robust}
\begin{figure}
\centerline{\includegraphics[width=0.45\textwidth]{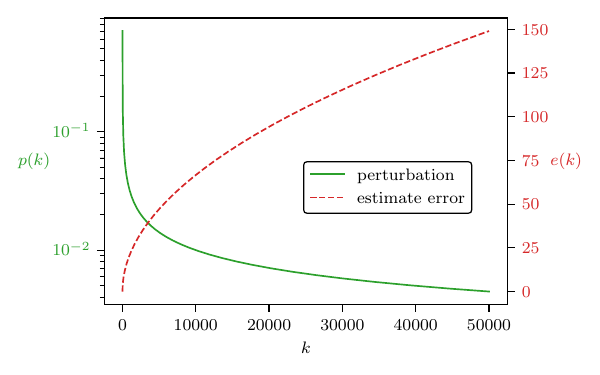}}
\caption{Converging perturbation $p(k)$ and diverging estimate error $e(k)$ for the
  time-varying Kalman filter applied to an unstabilizable system with poles on
  the unit circle.}
\label{fig:ex_4}
\end{figure}

Consider the unstabilizable case $A = 1, C=1, G=0, \Sigma_0 = 1, Q=1, R=1$.  Assume an
unmodeled disturbance, $p$, enters the system
\begin{equation*}
x^+ = A x + G w + p \qquad y = Cx + v
\end{equation*}
Let $v(k) = 0$ and $p(k) = 1/\sqrt{k+2}$ for $k \geq 0$.
Although the unmodeled disturbance is bounded, monotonically decreasing, and
converges to zero, and the time-varying Kalman filter is globally asymptotically stable,
the estimate error diverges even with perfect measurement.  \cref{fig:ex_4} shows
the result assuming perfect initial information, $\ehat(0) = 0$. 
It is well known that injecting a slowly decaying perturbation into an
asymptotically stable system may destabilize the system, and we see that
situation arises when the time-varying Kalman filter designed for an
unstabilizable system is applied to that system subjected to a slowly decaying 
disturbance.

 \section{Conclusions}
This paper has offered an alternative methodology to the standard Riccati analysis for understanding the asymptotic stability of the time-varying optimal estimator (Kalman filter) for a detectable but unstabilizable linear time invariant system.   This methodology is based on the properties of the optimal estimation problem that is equivalent to the time-varying Kalman filter. By proving Proposition \ref{prop:tvkf} with these optimization-based arguments, it was shown that the necessary and sufficient conditions for asymptotic stability (C1 and C2) can be readily deduced.

Part of the motivation for this work is to provide tools in the linear setting that are more easily extended to nonlinear systems as in moving horizon estimation. Since the standard Riccati iteration convergence analysis is notoriously cumbersome and does not extend in any obvious way to nonlinear systems, alternative analysis methods such as the one presented here might better prepare the next generation of graduate students for their future research careers.

Future work will be directed at deriving the necessary and sufficient conditions for exponential stability, which provides a robustness of the stability that is absent with asymptotic stability.
 
\bibliographystyle{abbrvnat}
{\small \bibliography{paper}

\section*{Acknowledgment}
This paper is dedicated to JBR's daughter and son-in-law, Melanie R. Rawlings and Zachary W. Zinda, at whose home much of this analysis was initiated. JBR would like to thank UCSB for a sabbatical in Spring 2025 that enabled this study to be undertaken.
TQ is grateful for research support from the National Science Foundation Graduate Research Fellowship Program under Grant No.~2139319.
MAM would like to thank Leibniz University Hannover for a sabbatical in Fall 2025 that he spent at UCSB.
The authors used Claude models, including Opus and Fable, from \cite{anthropic:claude}, and GPT-5.6 models from \cite{openai:gpt56:2026} for manuscript preparation and proof development. The proof-development assistance concerned \cref{prop:infhor}, the C2-necessity argument in \cref{prop:tvkf}, and the auxiliary facts on which that argument rests, namely \cref{eq:polysample,eq:gramian}.
\section{Proofs}
\label{sec:proofs}

The following preliminary facts will prove useful in the arguments to follow. We
collect them here mostly without proof in the interests of brevity.
\begin{enumerate}
  \item Let $V(x) = \tfrac12 x'Hx + x'd + f$ with $H > 0$. Then $\min_x V(x)$ has the unique solution $x^0 = -H^{-1}d$, and for all $x \in \bbR^n$
\begin{equation}
V(x) = V(x^0) + \tfrac12\norm{x - x^0}_H^2
\label{eq:quadmin}
\end{equation}
\item A semidefinite matrix $P \in \bbR^{n \times n}$, denoted, $P \geq 0$, has
 a factorization $P = U\Lambda U'$ in which  $U$ is orthogonal and $\Lambda$ is
 diagonal with real, non-negative elements. For all $x \in \bbR^n$
\begin{equation}
  \label{eq:quadbound}
\ulambda(P)  \norm{x}^2 \leq x'Px \leq \olambda(P) \norm{x}^2
\end{equation}
where $\ulambda(P), \olambda(P) \geq 0$ denote the smallest and largest eigenvalues of
$P$, respectively. 
\item The pseudoinverse of $P \geq 0$ is $P^\dagger = U \Lambda^\dagger U'$, in which
$\Lambda^\dagger$ is diagonal with $\Lambda^\dagger_{ii} = 1/\Lambda_{ii}$ if
$\Lambda_{ii} > 0$ and $\Lambda^\dagger_{ii} = 0$ otherwise, so that the range of $P^\dagger$ is the same as the range of $P$, i.e.,
$\mc{R}(P^\dagger) = \mc{R}(P)$. For all $x \in \mc{R}(P)$
\begin{equation}
\label{eq:rangebound}
\norm{x}^2 \leq \olambda(P) \, x' P^\dagger x
\end{equation}
\begin{proof}
\cref{eq:rangebound}, write $x = U \beta$, so that membership in
$\mc{R}(P)$ gives $\beta_i = 0$ whenever $\Lambda_{ii} = 0$. Since
$\Lambda_{ii} \leq \olambda(P)$ for every $i$,
\begin{equation*}
\olambda(P) \, x'P^\dagger x
 = \olambda(P) \sum_{\Lambda_{ii} > 0} \beta_i^2 / \Lambda_{ii}
 \geq \sum_{\Lambda_{ii} > 0} \beta_i^2 = \norm{x}^2
\end{equation*}
\end{proof}

\item An infinite horizon linear-quadratic problem for the system
$x^+ = Ax + Bu$, $x(0)=x_0$ with $(A,B)$ stabilizable and stage cost
$\ell_R(x,u)=\tfrac12(\norm{x}_{\tQ}^2 + \norm{u}_{\tR}^2)$ with
$\tQ \geq 0$, $\tR>0$, admits a stabilizing feedback law $u=Kx$, i.e.,
$A+BK$ is Schur, and this feedback has bounded infinite horizon cost
\begin{equation}
V_{R,\infty}^0 \eqbyd \sum_{j=0}^\infty \ell_R(x(j), u(j))
= \tfrac12 x_0' S x_0 \leq c_c \norm{x_0}^2
\label{eq:cc}
\end{equation}
for all $x_0 \in \bbR^n$, with $c_c = \tfrac12 \olambda(S)$, where $S \geq 0$
is the unique solution of the Lyapunov equation
\begin{equation*}
S - (A+BK)'S(A+BK) = \tQ + K'\tR K
\end{equation*}
which exists because $A+BK$ is Schur.

\item If $(A,C)$ is detectable, then for the linear system
$\tx^+ = A\tx + G\tw$, $\ty = C\tx$, there exists a quadratic IOSS-Lyapunov
function $V_\io(\tx)=\tfrac12\tx'P\tx$ with $P>0$ and constants
$a_1,a_2,a_3,c_1,c_2>0$ satisfying \citep{cai:teel:2008}
\begin{equation}
a_1\|\tx\|^2 \leq V_\io(\tx) \leq a_2\|\tx\|^2
\label{eq:iioss-bounds}
\end{equation}
and
\begin{equation}
V_\io(\tx^+) - V_\io(\tx)
\leq
-a_3\|\tx\|^2 + c_1\|\tw\|^2 + c_2\|\ty\|^2
\label{eq:iioss-dec}
\end{equation}

\item Let $q$ be a polynomial of degree at most $d$ with complex coefficients.
Then
\begin{gather}
\sum_{k=0}^{T-1}|q(k)|^2
\geq
\gamma_d\,T\,|q(0)|^2,
\quad
T\geq1
\label{eq:polysample}
\\
\gamma_d:=\frac{1}{2^{2d+2}(2d+1)^2} \notag
\end{gather}

\begin{proof}
Write $q=q_{\mathrm re}+iq_{\mathrm im}$ with real-coefficient
$q_{\mathrm re},q_{\mathrm im}$, and set
$p:=q_{\mathrm re}^2+q_{\mathrm im}^2$, a real polynomial of degree at most
$D:=2d$, nonnegative on $\mathbb R$, and equal to $|q|^2$ at every real
argument.

We compare $p(0)$ with samples of $p$ along strides. Fix $h\geq1$. Lagrange
interpolation at the $D+1$ nodes $\{h,2h,\dots,(D+1)h\}$ reproduces every
polynomial of degree at most $D$, so
\begin{gather*}
p(0)
=
\sum_{r=1}^{D+1}\beta_r\,p(rh)\\
\beta_r
=
\prod_{\substack{1\leq s\leq D+1\\s\ne r}}
\frac{0-sh}{rh-sh}
=
\prod_{\substack{1\leq s\leq D+1\\s\ne r}}
\frac{-s}{r-s}
\end{gather*}

Evaluating the product gives
$\beta_r=(-1)^{r-1}\binom{D+1}{r}$, so
$\sum_{r=1}^{D+1}|\beta_r|=2^{D+1}-1<2^{D+1}$. Since $p\geq0$, replacing each
coefficient by its magnitude can only increase the sum, and
\begin{equation}
p(0)
\leq
2^{D+1}\sum_{r=1}^{D+1}p(rh)
\label{eq:stride}
\end{equation}

Now sum \eqref{eq:stride} over the strides $h=1,\dots,H$, in which
$H:=\lfloor(T-1)/(D+1)\rfloor$ and $\lfloor\cdot\rfloor$ is the floor
function. Every term of the resulting double sum is $p(k)$ for an integer
$k=rh$ with $1\leq k\leq(D+1)H\leq T-1$, and each such $k$ arises for at most
one stride per $r$, hence at most $D+1$ times in all. The double sum is
therefore at most $(D+1)\sum_{k=0}^{T-1}p(k)$, and
\begin{equation}
\sum_{k=0}^{T-1}p(k)
\geq
\frac{H\,p(0)}{2^{D+1}(D+1)}
\label{eq:hcount}
\end{equation}

It remains to convert $H$ into $T$. For $T\geq2(D+2)$ we have
$2(T-D-2)\geq T$, so
$H\geq\frac{T-1}{D+1}-1=\frac{T-D-2}{D+1}
\geq\frac{T}{2(D+1)}$, and \eqref{eq:hcount} gives
$\sum_{k=0}^{T-1}p(k)\geq
T\,p(0)/\bigl(2^{D+2}(D+1)^2\bigr)$.
For $1\leq T<2(D+2)$ the $k=0$ term alone gives
$\sum_{k=0}^{T-1}p(k)\geq p(0)\geq T\,p(0)/\bigl(2(D+2)\bigr)$, which is the
stronger bound since
$2(D+2)\leq2^{D+2}(D+1)^2$. With $D=2d$, both ranges give
\eqref{eq:polysample}.
\end{proof}

\item Let $B$ be a complex square matrix with every eigenvalue satisfying
$|\lambda|\geq1$. Then there exists $c_B>0$, depending only on $B$, such that
\begin{equation}
W_T
:=
\sum_{k=0}^{T-1}(B^k)^*B^k
\succeq
c_B\,T I,
\qquad
T\geq1
\label{eq:gramian}
\end{equation}
In particular, $\ulambda(W_T)\to\infty$.

\begin{proof}
Write $B=\Xi J\Xi^{-1}$ in Jordan form and let $\usigma(\Xi)$ denote the smallest singular value of $\Xi$. 
Since $B^k=\Xi J^k\Xi^{-1}$,
for arbitrary $v$ we set $w:=\Xi^{-1}v$ and obtain
$\|B^kv\|=\|\Xi J^kw\|\geq\usigma(\Xi)\,\|J^kw\|$, so that
\[
\sum_{k=0}^{T-1}\|B^kv\|^2
\;\geq\;
\usigma(\Xi)^2\sum_{k=0}^{T-1}\|J^kw\|^2
\]
Partition $w=(w_1,\dots,w_p)$ conformally with the Jordan blocks
$J=\operatorname{diag}(J_1,\dots,J_p)$. The blocks act independently, so
$\|J^kw\|^2=\sum_{i=1}^{p}\|J_i^kw_i\|^2$, and the right side becomes
$\sum_{i=1}^{p}\sum_{k=0}^{T-1}\|J_i^kw_i\|^2$.
It therefore suffices to establish, for a single Jordan block of size~$m$ at
an eigenvalue $\lambda$ with $|\lambda|\geq1$, a constant $c>0$ such that
\[
\sum_{k=0}^{T-1}\|J^kz\|^2
\;\geq\;
c\,T\,\|z\|^2
\qquad\text{for all }z\text{ and }T\geq1.
\]

\emph{Blocks with $|\lambda|>1$.}
The inverse $J^{-1}$ has spectral radius $1/|\lambda|<1$, so its powers
converge to zero and $c_\lambda:=\sup_{k\geq0}\|J^{-k}\|$ is finite.
For every~$k$ we have
$\|z\|=\|J^{-k}J^kz\|\leq c_\lambda\|J^kz\|$, and summing from
$k=0$ to $T-1$ gives
\[
\sum_{k=0}^{T-1}\|J^kz\|^2
\;\geq\;
c_\lambda^{-2}\,T\,\|z\|^2
\]

\emph{Blocks with $|\lambda|=1$.}
Write $J=\lambda I+N$, where $N$ is the nilpotent shift with $N^m=0$.
The binomial theorem gives
\[
J^kz
\;=\;
\sum_{j=0}^{m-1}\binom{k}{j}\lambda^{k-j}N^jz
\]
Since $|\lambda|=1$, the $r$-th coordinate of $J^kz$ has modulus
$|q_r(k)|$, where
\[
q_r(k)
\;:=\;
\sum_{j=0}^{m-r}\binom{k}{j}\lambda^{-j}(N^jz)_r,
\quad r=1,\dots,m
\]
is a polynomial in~$k$ of degree at most $m-1$.
The upper limit is $m-r$ because $(N^jz)_r$ vanishes for $j>m-r$.
Since $q_r(0)=z_r$ and $\sum_{r=1}^{m}|q_r(0)|^2=\|z\|^2$, applying
\eqref{eq:polysample} to each coordinate and summing over~$r$ gives
\[
\sum_{k=0}^{T-1}\|J^kz\|^2
\;=\;
\sum_{r=1}^{m}\sum_{k=0}^{T-1}|q_r(k)|^2
\;\geq\;
\gamma_{m-1}\,T\,\|z\|^2
\]
so the constant for this block is $c=\gamma_{m-1}$.

\emph{Assembly.}
Writing $\underline c_J:=\min_i c_i$ for the smallest block constant and
summing the per-block bounds gives
$\sum_{k=0}^{T-1}\|J^kw\|^2\geq\underline c_J\,T\,\|w\|^2$.
Since $\|w\|=\|\Xi^{-1}v\|\geq\|v\|/\|\Xi\|$, we obtain
\eqref{eq:gramian} with
$c_B=\usigma(\Xi)^2\,\underline c_J/\|\Xi\|^2$.
\end{proof}
\end{enumerate} 

\noindent\paragraph{Uniform upper bound of $V_T^0$ for semidefinite $\Sigma_0$ in \cref{lem:unibounded}.}
\begin{proof}
Here we finish the deferred proof of \cref{lem:unibounded} by finding a uniform upper bound for $V_T^0$ with $\Sigma_0 \geq 0$ under assumption C2. Assume without 
loss of generality that the  unstabilizable system is partitioned into stabilizability canonical form as in \cref{eq:contrb}.
First we establish that all $\chi_2(0) \in \bbR^{n_2}$ are feasible for problem $\bbP_T$. Recall that the columns of $U_2$ are a basis for the null space of $\Sigma_0$ from the SVD in \cref{eq:SigmaSVD}. Partition $U_2$ and $\ox_0$ conformably with \cref{eq:contrb} as $U_2 = [\,U_{12}'\ U_{22}'\,]',\, \ox_0=(\ox_{10}, \ox_{20})$ and define
\begin{equation*}
  E_2 \eqbyd \begin{bmatrix} 0 \\ I_{n_2} \end{bmatrix}
\end{equation*}
whose columns are a basis for the unstable, uncontrollable modes of $(A,G)$ in the coordinates of \cref{eq:contrb}. Let $\xi \in \bbR^n$ belong to the intersection of these two subspaces, so that $\xi=U_2 \alpha = E_2 \beta$, or
\begin{equation*}
  0 = \begin{bmatrix} U_2 & E_2 \end{bmatrix} \begin{bmatrix} \alpha \\ -\beta \end{bmatrix} =
\begin{bmatrix} U_{12} & 0 \\ U_{22} & I \end{bmatrix} \begin{bmatrix} \alpha \\ -\beta \end{bmatrix}.
\end{equation*}
Under C2, $\xi=0$ is the only element of the intersection, which requires $\alpha=0, \beta=0$ as the only solution, and hence the columns of $\begin{bmatrix} U_2 & E_2 \end{bmatrix}$ are linearly independent. Given the $0$ and $I$ blocks in $E_2$, these columns are linearly independent if and only if the columns of $U_{12}$ are linearly independent. The constraint on the initial condition in $\bbP_T$ is $U_2'(\chi(0)-\ox_0) =0$, or in partitioned form
\begin{equation*}
  \begin{bmatrix} U_{12}' & U_{22}' \end{bmatrix}
  \begin{bmatrix} \chi_1(0) - \ox_{10} \\ \chi_2(0) - \ox_{20} \end{bmatrix} = 0
\end{equation*}
or
\begin{equation*}
U_{12}' (\chi_1(0) - \ox_{10})  = -U_{22}' (\chi_2(0) - \ox_{20} )
\end{equation*}
Since $U_{12}$ has linearly independent columns, $U_{12}'$ has linearly independent rows, and this constraint has a solution $\chi_1(0)$ for all $\chi_2(0) \in \bbR^{n_2}$ and $\ox_0 \in \bbR^n$. Therefore all $\chi_2(0)$ are feasible in $\bbP_T$.  

In particular $\chi_2(0) = x_2(0)$ is feasible in $\bbP_T$, and we choose $\chi_1(0) = \ox_{10} - (U_{12}')^\dagger U_{22}'(x_2(0)-\ox_{20})$, which satisfies the constraint. 
Defining the candidate error  $e \eqbyd x - \chi$ gives the candidate error system $e^+ = Ae - G \omega$, $\nu = Ce$. Since $e_2(0)= x_2(0)-\chi_2(0) = 0$, we have $e_2(j) = 0$, all $0 \leq j \leq T-1$, which gives the remaining system
\begin{equation*}
  e_1^+  = A_1 e_1 - G_1 \omega \qquad \nu = C_1 e_1
\end{equation*}
with $e_1(0) = x_1(0) - \chi_1(0)$. With $\chi_1(0)$ specified, the optimization over $\omega$ in problem $\bbP_T$ becomes a linear-quadratic problem for the system $(A_1, G_1)$
with state weight $C_1'R^{-1}C_1 \geq 0$ and input weight $Q^{-1} > 0$.
Since $(A_1,G_1)$ is stabilizable, the bound \cref{eq:cc} applies. Problem $\bbP_T$ sums only $T$ stages, but its stage costs are nonnegative, so the optimal $T$-stage $\omega$ cost is bounded above by the first $T$ stage costs of the stabilizing feedback of \cref{eq:cc}, and therefore by that feedback's infinite horizon cost, for every $T \geq 0$. Applying the bound in \cref{eq:cc},
\begin{align*}
  V_{R,\infty}^0 &\leq c_c \norm{e_1(0)}^2 \\
  &\qquad= c_c \norm{x_1(0) - \ox_{10} + (U_{12}')^\dagger U_{22}'(x_2(0)-\ox_{20})}^2 \\
  &\qquad= c_c \norm{ \begin{bmatrix} I & (U_{12}')^\dagger U_{22}'\end{bmatrix} (x(0) - \ox_0) }^2 \\
  &\qquad\leq c_d \norm{x(0) - \ox_0}^ 2
\end{align*}
with $c_d = c_c \norm{ \begin{bmatrix} I & (U_{12}')^\dagger U_{22}'\end{bmatrix} }^2$. 
The cost of the initial stage is zero if $\Sigma_0=0$; otherwise,  $\tSigma_0 > 0$ and the cost of the initial stage has an upper bound
\begin{align*}
  \ell_x(\chi(0) - \ox_0) &= \frac{1}{2}\norm{ \begin{bmatrix} -(U_{12}')^\dagger U_{22}'(x_2(0)-\ox_{20}) \\ x_2(0) -\ox_{20} \end{bmatrix} }_{\Sigma_0^\dagger}^2 \\
  &\leq \frac{1}{2\ulambda(\tSigma_0)} \norm{ \begin{bmatrix} -(U_{12}')^\dagger U_{22}'\\ I  \end{bmatrix}}^2 \norm{x_2(0) -\ox_{20}}^2\\
&\leq c_e \norm{x_2(0) -\ox_{20}}^2
\end{align*}
with $c_e = \frac{1}{2\ulambda(\tSigma_0)} \norm{ \begin{bmatrix} -(U_{12}')^\dagger U_{22}' \\ I  \end{bmatrix}}^2 $, and $c_e \eqbyd 0$ if $\Sigma_0 = 0$.

Adding these two costs gives the uniform upper bound for $\bbP_T$
\begin{equation*}
  V_T^0 \leq c_v \norm{x(0) - \ox_0}^2  \qquad c_v \eqbyd c_d + c_e
\end{equation*}
which holds for all $T \geq 0$.
\end{proof}

\noindent\paragraph{Existence of infinite horizon solution.}
\nopagebreak

We now establish the following proposition that is used in the proof of \cref{prop:modQgas}.
 \begin{proposition}[Infinite horizon state estimation with C1, C2 and $\Sigma_0\geq 0$]
\label{prop:infhor}
Consider state estimation problem $\bbP_T$ in \cref{lem:semiPT} with $\Sigma_0\geq 0$. \Cref{it:zlim,it:Vlim} hold under assumption C2 alone; \cref{it:xTT} additionally requires detectability (C1).
\begin{enumerate}

\item
  \label{it:zlim}
The following limits exist
\begin{align*}
  \xhat(0 \mid \infty) &\eqbyd \lim_{T \rightarrow \infty} \xhat(0 \mid T) \\
  \what(j \mid \infty) &\eqbyd \lim_{T \rightarrow \infty} \what(j \mid T)
\end{align*}
the second for every fixed $j \geq 0$, with the limit taken over $T \geq j+1$.
Define $\whatseq_\infty \eqbyd (\what(j\mid\infty))_{j\geq 0}$.

\item
\label{it:xTT}
  $\xhat(T \mid T) - \xhat(T \mid \infty) \rightarrow 0$ as $T \rightarrow \infty$, where $\xhat(\cdot \mid \infty)$ is the trajectory generated by $(\xhat(0 \mid \infty), \whatseq_\infty)$.

\item
\label{it:Vlim}
The infinite horizon cost function defined by
\begin{gather*}
V_\infty(\chi(0), \omegaseq_\infty) \eqbyd \tfrac12\norm{\chi(0)-\ox_0}_{\Sigma_0^{\dagger}}^2 \\
\qquad + \tfrac12\sum_{j=0}^\infty \big(\norm{\omega(j)}_{Q^{-1}}^2 + \norm{\nu(j)}_{R^{-1}}^2\big) \\
\text{subject to~} \chi^+ = A \chi + G \omega, \quad y = C\chi + \nu,\\
U_2'(\chi(0)-\ox_0) = 0
\end{gather*}
has the unique minimizer $(\xhat(0\mid\infty), \whatseq_\infty)$ subject to its constraints, and
\begin{equation*}
\min V_\infty = V_\infty(\xhat( 0\mid \infty), \whatseq_\infty) = V^0_\infty
\end{equation*}
where $V^0_\infty \eqbyd \lim_{T\rightarrow \infty} V^0_T$.
\end{enumerate}
\end{proposition}
\begin{proof}
From \cref{lem:exist} the optimizer $(\xhat(0\mid T), \whatseq_T)$ of $\bbP_T$ exists for all $T \geq 0$ with $V^0_T = V_T(\xhat(0\mid T), \whatseq_T)$. Set $\chi(0)-\ox_0 = U \alpha = U_1\alpha_1$, since the constraint renders $\alpha_2 = 0$ (for $\Sigma_0=0$, $U_1$ is vacuous, $\chi(0)=\ox_0$, and every $\alpha_1$ term below is absent). With $z \eqbyd (\alpha_1, \omegaseq_T)$,
\begin{equation*}
  V_T(z) = \tfrac12 z'H_Tz + z'd_T + f_T
\end{equation*}
Here $d_T$ and $f_T$ collect the linear and constant terms, respectively; their explicit expressions are not needed. The Hessian $H_T$ is the sum of the block diagonal $\diag(\tSigma_0^{-1}, Q^{-1}, \ldots, Q^{-1})$ contributed by the prior and process terms, and the matrix representing the second-degree part of the output term $\tfrac12\sum_{j=0}^{T-1}\norm{\nu(j)}_{R^{-1}}^2$, in which $\chi$ is generated by $z$ and $\nu(j) = y(j)-C\chi(j)$. Since $\chi$ depends affinely on $z$, each residual $\nu(j)$ is affine in $z$, so each $\norm{\nu(j)}_{R^{-1}}^2$ is quadratic in $z$ with positive semidefinite second-degree part, the remaining terms being linear and constant in $z$. Since $\tSigma_0 > 0$ and $Q > 0$, the block diagonal is positive definite and hence $H_T > 0$, so the minimizer $\hat z_T \eqbyd (\hat\alpha_{1,T}, \whatseq_T)$ exists and is unique; its first component defines $\hat\alpha_{1,T}$, so that $\xhat(0\mid T) = \ox_0 + U_1\hat\alpha_{1,T}$. For every $z$ write $\delta z \eqbyd z - \hat z_T$, with components $\delta\alpha_1 \eqbyd \alpha_1-\hat\alpha_{1,T}$ and $\delta\omegaseq_T \eqbyd \omegaseq_T-\whatseq_T$. Then \cref{eq:quadmin}, expanded along the two contributions to $H_T$, gives
\begin{equation}
\begin{aligned}
  V_T(z)&-V_T^0 = \tfrac12\norm{\delta z}_{H_T}^2 \\
  &= \tfrac12\left(\norm{\delta\alpha_1}_{\tSigma_0^{-1}}^2 +\sum_{j=0}^{T-1}\big(\norm{\delta\omega(j)}_{Q^{-1}}^2+\norm{\delta\nu(j)}_{R^{-1}}^2\big)\right)
\end{aligned}
\label{eq:gap}
\end{equation}
with $\delta\nu(j) \eqbyd -C\delta\chi(j)$, $\delta\chi(0) = U_1\delta\alpha_1$, and $\delta\chi^+ = A\delta\chi+G\delta\omega$.

For horizons $p > m$, the truncation $\tz_{m\mid p} \eqbyd (\hat\alpha_{1,p}, \what(0\mid p), \ldots, \what(m-1\mid p))$ is feasible for $\bbP_m$ with $V_m(\tz_{m\mid p}) = V_p^0-\sum_{j=m}^{p-1}\ell(\what(j\mid p), \vhat(j\mid p)) \leq V_p^0$. Since $\ell(\cdot) \geq 0$, $(V_T^0)_{T\geq0}$ is nondecreasing; bounded above by \cref{lem:unibounded}, it converges to $V_\infty^0$, and for every $\eps > 0$ there is $T_\eps$ with $V_p^0-V_m^0 \leq \eps$ for all $p > m \geq T_\eps$. Applying \cref{eq:gap} at horizon $m$ to $z = \tz_{m\mid p}$ and writing the optimizer differences as
$
\delta\hat\alpha_{1|_m^p} \eqbyd \hat\alpha_{1,p}-\hat\alpha_{1,m}, \quad
\delta\what(j\mid_m^p) \eqbyd \what(j\mid p)-\what(j\mid m), \quad
\delta\vhat(j\mid_m^p) \eqbyd \vhat(j\mid p)-\vhat(j\mid m)
$
gives
\begin{equation}
\begin{aligned}
  &\tfrac12\left(\norm{\delta\hat\alpha_{1|_m^p}}_{\tSigma_0^{-1}}^2  +\sum_{j=0}^{m-1}\big(\norm{\delta\what(j\mid_m^p)}_{Q^{-1}}^2+\norm{\delta\vhat(j\mid_m^p)}_{R^{-1}}^2\big)\right) \\
  &\quad = V_m(\tz_{m\mid p})-V_m^0 \leq \eps
\end{aligned}
\label{eq:gapopt}
\end{equation}
Each term is bounded by the sum; with $\ulambda(\tSigma_0^{-1}) > 0$ and $\norm{U_1\delta\hat\alpha_{1|_m^p}} = \norm{\delta\hat\alpha_{1|_m^p}}$, the sequences $(\xhat(0\mid T))_{T\geq0}$ and, for every fixed $j\geq0$, $(\what(j\mid T))_{T\geq j+1}$ are Cauchy, hence converge, establishing \cref{it:zlim}. Denote the limits $\xhat(0\mid\infty)$ and $\whatseq_\infty$. By \cref{it:zlim} and linearity, for every fixed $j$ the optimal trajectory converges pointwise, $\xhat(j\mid T) \rightarrow \xhat(j\mid\infty)$, in which $\xhat(\cdot\mid\infty)$ is the trajectory generated by $(\xhat(0\mid\infty), \whatseq_\infty)$. The residuals converge as well, with $\vhat(j\mid\infty) \eqbyd y(j)-C\xhat(j\mid\infty) = \lim_{T\rightarrow\infty}\vhat(j\mid T)$.

Passing $p\rightarrow\infty$ in \cref{eq:gapopt}---each summand converges pointwise by \cref{it:zlim}---gives, for all $T \geq T_\eps$,
\begin{equation*}
  \norm{\delta\hat\alpha_{1|_T^\infty}}_{\tSigma_0^{-1}}^2 + \sum_{j=0}^{T-1}\big(\norm{\delta\what(j\mid_T^\infty)}_{Q^{-1}}^2+\norm{\delta\vhat(j\mid_T^\infty)}_{R^{-1}}^2\big) \leq 2\eps
\end{equation*}
in which $\hat\alpha_{1,\infty} \eqbyd U_1'(\xhat(0\mid\infty)-\ox_0)$ is the limit of $\hat\alpha_{1,T}$ by \cref{it:zlim}. By the eigenvalue bounds for $\tSigma_0, Q, R$ and $\norm{U_1\delta\hat\alpha_{1|_T^\infty}} = \norm{\delta\hat\alpha_{1|_T^\infty}}$, there is a $c_z > 0$ with
\begin{equation*}
  \norm{\delta\xhat(0\mid_T^\infty)}^2+\sum_{j=0}^{T-1}\norm{\delta\what(j\mid_T^\infty)}^2+\sum_{j=0}^{T-1}\norm{\delta\vhat(j\mid_T^\infty)}^2 \leq c_z\eps
\end{equation*}
Since the system is detectable (C1), the difference $\delta\xhat(j\mid_T^\infty)$ satisfies \eqref{eq:iioss-dec} with $\tx(j)=\delta\xhat(j\mid_T^\infty)$, $\tw(j)=\delta\what(j\mid_T^\infty)$, and $\ty(j)=-\delta\vhat(j\mid_T^\infty)$. Summing \eqref{eq:iioss-dec} over $j=0,\ldots,T-1$, discarding the nonpositive $-a_3$ terms, and applying \eqref{eq:iioss-bounds} on both ends gives
\begin{align*}
a_1\norm{\delta\xhat(T\mid_T^\infty)}^2
&\leq
a_2\norm{\delta\xhat(0\mid_T^\infty)}^2\\
&\phantom{\leq}
+ \sum_{j=0}^{T-1}c_1\norm{\delta\what(j\mid_T^\infty)}^2
+ c_2 \norm{\delta\vhat(j\mid_T^\infty)}^2 \\
&\leq
\max(a_2,c_1,c_2)\,c_z\eps
\end{align*}
for all $T \geq T_\eps$. Given $\eps' > 0$, choosing
$\eps \leq a_1\eps'^2/(\max(a_2,c_1,c_2)\,c_z)$ gives
$\norm{\xhat(T\mid T)-\xhat(T\mid\infty)} = \norm{\delta\xhat(T\mid_T^\infty)} \leq \eps'$
for all $T \geq T_\eps$, establishing \cref{it:xTT}.

Finally, by definition of $V_T$,
\begin{equation*}
  2V_T^0 = \norm{\hat\alpha_{1,T}}_{\tSigma_0^{-1}}^2 + \sum_{j=0}^{T-1}\big(\norm{\what(j\mid T)}_{Q^{-1}}^2+\norm{\vhat(j\mid T)}_{R^{-1}}^2\big)
\end{equation*}
For finite $N$ and $T \geq N+1$, dropping the tail gives $\sum_{j=0}^{N-1}(\norm{\what(j\mid T)}_{Q^{-1}}^2+\norm{\vhat(j\mid T)}_{R^{-1}}^2) \leq 2V_T^0-\norm{\hat\alpha_{1,T}}_{\tSigma_0^{-1}}^2$. Letting $T\rightarrow\infty$ (pointwise, using the convergence of $\what(j\mid T)$ and $\vhat(j\mid T)$ established above, with $\norm{\hat\alpha_{1,T}} \rightarrow \norm{\hat\alpha_{1,\infty}}$) and then $N\rightarrow\infty$ (monotone convergence),
\begin{equation*}
  \norm{\hat\alpha_{1,\infty}}_{\tSigma_0^{-1}}^2 + \sum_{j=0}^\infty\big(\norm{\what(j\mid\infty)}_{Q^{-1}}^2+\norm{\vhat(j\mid\infty)}_{R^{-1}}^2\big) \leq 2V_\infty^0
\end{equation*}
Since $\Sigma_0^\dagger = U_1\tSigma_0^{-1}U_1'$ gives $\norm{\xhat(0\mid\infty)-\ox_0}_{\Sigma_0^\dagger}^2 = \norm{\hat\alpha_{1,\infty}}_{\tSigma_0^{-1}}^2$, this is $V_\infty(\xhat(0\mid\infty), \whatseq_\infty) \leq V_\infty^0$.

Conversely, $U_2'(\xhat(0\mid T)-\ox_0) = 0$ for every $T$ and
$\xhat(0\mid T) \rightarrow \xhat(0\mid\infty)$ by \cref{it:zlim}, so
$U_2'(\xhat(0\mid\infty)-\ox_0) = 0$ and
$(\xhat(0\mid\infty), \whatseq_\infty[0:T-1])$ is feasible for $\bbP_T$; hence
$V_T(\xhat(0\mid\infty), \whatseq_\infty[0:T-1]) \geq V_T^0$. The left side increases to $V_\infty(\xhat(0\mid\infty), \whatseq_\infty)$ and the right to $V_\infty^0$, so $V_\infty(\xhat(0\mid\infty), \whatseq_\infty) \geq V_\infty^0$. Hence $V_\infty(\xhat(0\mid\infty), \whatseq_\infty) = V_\infty^0$. The same truncation argument applies to every pair satisfying the constraints of $V_\infty$, so $V_\infty \geq V_\infty^0$ on its feasible set, and the limit pair is a minimizer. As in \cref{lem:exist}, $V_\infty$ is strictly convex in $(\alpha_1, \omegaseq_\infty)$ on this set, since $\tSigma_0 > 0$ and $Q > 0$, so the minimizer is unique. This establishes \cref{it:Vlim}, which completes the proof.
\end{proof}

\noindent \paragraph{Proof of \cref{prop:tvkfQuns}.}
\begin{proof}
Here we establish \cref{eq:QunsInitUB}, \cref{eq:QunsLBUB}, and  \cref{eq:QunsDecrease}.
To build the $Q(j\mid k)$ function with cost decrease rather than cost increase as in $V_T^0$, define the partial costs
\begin{equation*}
V^0(j \mid k) \eqbyd \ell_x(\xhat(0 \mid k) -\ox_0) + \sum_{i=0}^{j-1} \ell(\what(i \mid k), \vhat(i \mid k))
\end{equation*}
for $j \leq k$, and flip the cost via $Z(j\mid k) \eqbyd V_\infty^0 - V^0(j \mid k)$,
where $V_\infty^0$ exists by \cref{lem:unibounded} under C2. With this definition we have the properties
\begin{align*}
  0 &\leq Z(j \mid k) \\
Z(j+1 \mid k) &=  Z(j \mid k) - \ell(\what(j \mid k), \vhat(j \mid k))
\end{align*}
Notice that $Z(j\mid k)$ does not have its usual upper bound since the system is not necessarily stabilizable. Next, by detectability (C1), there exists a quadratic IOSS-Lyapunov function
$V_\io(\tx)=\tfrac12\tx'P\tx$, with $P\succ0$ and
$a_1,a_2,a_3,c_1,c_2>0$ satisfying
\eqref{eq:iioss-bounds}--\eqref{eq:iioss-dec}, and we can rescale $V_\io$ by
any scalar $\rho > 0$ to rescale its constants. Under \cref{eq:nommeas}, the
difference $\tx(j) \eqbyd \xhat(j\mid k) - x(j)$ satisfies 
\eqref{eq:iioss-dec} with $\tw(j) = \what(j\mid k)$ and
$\ty(j) = -\vhat(j\mid k)$. Defining
\begin{equation*}
Q(j \mid k) \eqbyd Z(j \mid k) + V_\io(\xhat(j \mid k) - x(j))
\end{equation*}
gives $Q$ with these properties 
\begin{align*}
a_1\norm{x(j)-\xhat(j \mid k)}^2 & \leq Q(j \mid k)  \\
Q(j+1\mid k) &\leq Q(j\mid k) - a_3\norm{x(j)-\xhat(j\mid k)}^2
\end{align*}
where we have rescaled $V_\io$ if necessary so that $0 < c_1 \leq \tfrac12 \ulambda(Q^{-1})$ and $0< c_2 \leq \tfrac12 \ulambda(R^{-1})$.
Note that we have established the lower bound in \cref{eq:QunsLBUB}, \cref{eq:QunsDecrease}, and require  only \cref{eq:QunsInitUB}, the upper bound for $Q(0 \mid k)$, to complete the proof.

Towards that end, we start with
\begin{align*}
Q(0 \mid k)
&\eqbyd V_\infty^0 - V^0(0 \mid k) + V_\io(\xhat(0\mid k)-x(0)) \\
&\leq c_v \norm{x(0) - \ox_0}^2
   - \tfrac12\norm{\xhat(0 \mid k)-\ox_0}_{\Sigma_0^\dagger}^2 \\
&\qquad + \tfrac12\norm{\xhat(0\mid k)-x(0)}_{P}^2 \\
&\leq c_v \norm{x(0) - \ox_0}^2
   - \tfrac12\norm{\xhat(0 \mid k)-\ox_0}_{\Sigma_0^\dagger}^2 \\
&\qquad + \olambda(P)
   \bigl(\norm{\xhat(0\mid k)- \ox_0}^2 + \norm{\ox_0 - x(0)}^2\bigr) \\
&= c_0 \norm{x(0) - \ox_0}^2
   + \olambda(P) \norm{\xhat(0\mid k)- \ox_0}^2 \\
&\qquad - \tfrac12\norm{\xhat(0 \mid k)-\ox_0}_{\Sigma_0^\dagger}^2
\end{align*}
where $c_0 \eqbyd c_v + \olambda(P)$, which is positive since $P \succ 0$, and we have used \cref{lem:unibounded}, which requires C2, and \cref{eq:quadbound} and the fact that $\norm{a + b}^2 \leq 2 (\norm{a}^2 + \norm{b}^2)$ for $a, b \in \bbR^n$.

We would like to rescale the IOSS-Lyapunov function $V_\io$ again if necessary, by rescaling matrix $P$, so that $\olambda(P) \norm{\xhat(0\mid k)- \ox_0}^2 \leq \tfrac12\norm{\xhat(0 \mid k)-\ox_0}_{\Sigma_0^\dagger}^2$, but recall that $\Sigma_0^\dagger$ is only semidefinite so that is not possible for arbitrary vectors. Instead we again use the SVD of $\Sigma_0$ to express
$\xhat(0 \mid k)-\ox_0 = U_1 \alpha_1 + U_2 \alpha_2$, which gives $\alpha_2 = 0$
because of the constraint in $\bbP_k$, so $\xhat(0 \mid k) - \ox_0$ lies in
$\mc{R}(\Sigma_0)$. Rescaling $V_\io$ if necessary so that
$2\olambda(P) \olambda(\Sigma_0) \leq 1$ and applying \cref{eq:rangebound} gives
\begin{equation*}
\olambda(P) \norm{\xhat(0\mid k)- \ox_0}^2
  - \tfrac12\norm{\xhat(0 \mid k)-\ox_0}_{\Sigma_0^\dagger}^2 \leq 0
\end{equation*}
for all feasible $\xhat(0 \mid k)$ in $\bbP_k$.
Using this fact in the previous inequality for $Q(0 \mid k)$ gives
\begin{equation*}
Q(0 \mid k) \leq c_0 \norm{x(0) - \ox_0}^2
\end{equation*}
verifying \cref{eq:QunsInitUB}.
We have therefore established \cref{eq:QunsInitUB}, the lower bound in \cref{eq:QunsLBUB}, and \cref{eq:QunsDecrease}.  Notice that we have produced quadratic $\mc{K}_\infty$-functions verifying the final claim in the proposition.
\end{proof}
\end{document}